\documentclass[10pt]{article}

\usepackage{amssymb,amsfonts,amsmath,graphicx}

\newtheorem{corollary}{Corollary}{}
\newtheorem{theorem}{Theorem}{}
\newtheorem{lemma}{Lemma}{}
{} 
\newtheorem{definition}{Definition}{}
\newenvironment{proof}{\textbf{Proof}}{}

\newcommand{\qed}{\ensuremath{\hfill\Box}}
\newcommand{\OPT}{\textsc{Opt}}

\newcommand{\SUB}{\textsc{Sub}}
\newcommand{\ALG}{\textsc{Alg}}
\newcommand{\Alg}{\textsc{Alg}}
\newcommand{\alg}{\textsc{Alg}}

\newcommand{\expected}{\mathbb{E}}
\newcommand{\RR}{\mathbb{R}}
\newcommand{\OO}{\tilde{O}}

\newcommand{\TT}{\mathcal{T}}
\newcommand{\NP}{\ensuremath{\mathcal{NP}}}
\newcommand{\eps}{\epsilon}

\newcommand{\TRP}{\text{Traveling Repairman Problem}}

\begin{document}

\title{\Large Polynomial time approximation schemes for the traveling repairman 
 and other minimum\\ latency problems.}
\author{Ren\'e Sitters \thanks{ Vrije Universiteit, Amsterdam. Department of Econometrics and Operations Research.} }
\date{}

\maketitle

\begin{abstract} 
\small\baselineskip=9pt 
We give a polynomial time, $(1+\epsilon)$-approximation algorithm for the traveling repairman problem (TRP) in the Euclidean plane and on weighted trees. This improves on the known quasi-polynomial time approximation schemes for these problems.
The algorithm is based on a simple technique that reduces the TRP to what we call the \emph{segmented TSP}. Here, we are given numbers $l_1,\dots,l_K$ and $n_1,\dots,n_K$ and we need to find a path that visits at least $n_h$ points within path distance $l_h$ from the starting point for all $h\in\{1,\dots,K\}$.
A solution is $\alpha$-approximate if at least $n_h$ points are visited within distance $\alpha l_h$. It is shown that any algorithm that is $\alpha$-approximate for \emph{every constant} $K$ in some metric space, gives an $\alpha(1+\epsilon)$-approximation for the TRP in the same metric space. Subsequently, approximation schemes are given for this segmented TSP problem in the plane and on weighted trees.
The segmented TSP with only one segment ($K=1$) is equivalent to the $k$-TSP for which a $(2+\epsilon)$-approximation is known for a general metric space. Hence, this approach through the segmented TSP gives new impulse for improving on the 3.59-approximation for TRP in a general metric space.
A similar reduction applies to many other minimum latency problems.
To illustrate the strength of this approach we apply it to the well-studied scheduling problem of minimizing total weighted completion time under precedence constraints, $1|prec|\sum w_{j}C_{j}$, and present a polynomial time approximation scheme for the case of interval order precedence constraints. This improves on the known $3/2$-approximation for this problem.
Both approximation schemes apply as well if release dates are added to the problem. \\

Remark: In the previous version of this report it was claimed that the PTAS applies to TRP on weighted planar graphs too. This turned out to be incorrect. Getting a PTAS or even quasi-PTAS for the TRP on weighted planar graphs remains an open problem. See the discussion on open problems in this report.
\end{abstract}
 
\section{Introduction} 
The traveling repairman problem (TRP) (also known as the minimum latency problem) is similar to the well-known traveling salesman problem (TSP). An instance is given by points in a metric space and a feasible solution is a path $\Pi$, starting at a given origin $r$, that visits each of the points. The \emph{completion time} of a point $v$ is the distance from $r$ to $v$ on path $\Pi$ and the objective is to minimize the total completion time of the points. Hence, the problem can be seen as a traveling repairman who's aim is to minimize the average arrival time at the  clients. The traveling salesman on the other hand aims at minimizing the travel time of the salesman himself.   
 
The approximability of the TRP has been the subject of many papers~\cite{Afrati:86,ArcherBlasiak2010,ArcherWillamson:2003,AroraKara2003,Blum:94,ChaudhuriGRT2003,DewildeCCSV2013,FakchHR2003,GarciaJT:2002,GoemansKleinberg:98,KrumkePPS2003,NagarajanRavi2008,Sitters:02,Tsitsiklis:1992}. The first approximation ratio, given by Blum et al.~\cite{Blum:94}, was $144$ and the current smallest ratio for general metrics is $3.59$ due to  Chaudhuri et al.~\cite{ChaudhuriGRT2003}.
Even for trees, the best polynomial time approximation ratio was $3.59$ until recently Archer and Blasiak reduced it to 3.03~\cite{ArcherBlasiak2010}. $\NP$-hardness of the tree case was shown in~\cite{Sitters:02}. For the Euclidean plane, a $3.59$-approximation follows from the TRP algorithm by Goemans and Kleinberg~\cite{GoemansKleinberg:98} in combination with the polynomial time approximation scheme (PTAS) for the Euclidean $k$-TSP by Arora~\cite{Arora1998}. A quasi-polynomial time approximation scheme (QPTAS) for the TRP on trees and the Euclidean plane was given by Arora and Karakostas~\cite{AroraKara:1999,AroraKara2003}. Arora and Karakostas  write that they `do not know whether the running time can be reduced to polynomial'. Here we show that this is indeed possible. The authors also remark that their QPTAS applies to weighted planar graphs too. However, this claim turned out to be incorrect~\cite{KarakostasCommunication2014}. Hence, getting a PTAS or even en QPTAS is still an open problem.

All known TRP algorithms solve some form of the $k$-TSP or $k$-MST as a subroutine. In the $k$-TSP ($k$-MST), one needs to find the shortest tour (tree) that visits at least $k$ of the $n$ input points.
For example, the algorithm by Goemans and Kleinberg~\cite{GoemansKleinberg:98} first computes approximate $k$-TSP tours for all $k\leqslant n$ and then combines a subset of these tours into one TRP solution. An alternative approach is to make subtours of geometrically increasing length and to visit a maximum number of points in each subtour. The obtained ratio is $3.59\alpha$, where $\alpha$ is the approximation ratio of the $k$-TSP (or $k$-MST). So far, the best ratio for $k$-TSP and $k$-MST is $2+\epsilon$~\cite{AroraKarakostas:2000}. Chaudhuri et al. \cite{ChaudhuriGRT2003} found  a way to bypass this factor 2 in the analysis and noted that  breaking the barrier of 3.59 would probably involve an approach different than combining small tours. Indeed, the exact algorithm for TRP on the line~\cite{Afrati:86} and the quasi-PTAS~\cite{AroraKara2003} for the plane and trees are different since they find a solution directly by one dynamic program.
Our approach here is to combine both ideas. Instead of solving $n$ times a $k$-TSP, we solve a polynomial number of relatively large subtour problems which we call the \emph{segmented TSP}. Dynamic programming is applied to combine a subset of these subtours into one solution. As in~\cite{GoemansKleinberg:98}, the subtours are of geometrically  increasing length. However, the multiplication factor is not an absolute constant but a large constant depending on $\epsilon$. This way, we loose only a $1+\epsilon$ factor due to the returns to the origin. The downside is that the number of subproblems increases as well as their complexity. We show that it is enough for a PTAS to solve only a polynomial number of these subproblems approximately. More precesily, it is shown that any $\alpha$-approximation algorithm for the segmented TSP gives an $\alpha(1+\epsilon)$-approximation for the TRP in the same metric space. Subsequently, we give a PTAS for the segmented TSP on weighted trees and the Euclidean plane. An intersting by-product is that this gives a new direction for improving on the general $3.59$ approximation since any approximation ratio better than $3.59$ for the segmented TSP gives an improved factor for the TRP! The segmented TSP has not been studied before and no constant factor approximation is known for general metric spaces. It seems unlikely though that the factor 3.59 will show up in the analysis of the segmented TSP.

The same approach can be applied to many sequencing problem with minimum total (weighted) completion time objective. This is illustrated by our second application, which is the notorious scheduling problem of minimizing total weighted completion time under precedence constraints, known as $1|prec|\sum w_{j}C_{j}$ in the standard scheduling notation. (See e.g. Graham et al.~\cite{GrahamLLR1979}.) The approximability of this problem has been studied in many papers. The problem is known to be \NP-hard~\cite{Lawler1978,LenstraRK1978} and several 2-approximation algorithms are known. The paper~\cite{AmbuhlMMS2011} gives a recent overview on the status of this problem.
We present a polynomial time approximation scheme for the case of interval order precedence constraints.
Woeginger~\cite{Woeginger2003} gave a 1.62-approximation algorithm and a $3/2$-approximation was given by Amb\"uhl et al~\cite{AmbuhlMMS2011}. Somewhat surprisingly, the same paper shows that scheduling interval orders is in fact \NP-hard. Hence, our PTAS closes the gap in the approximabilty for this problem.

\subsection{Preliminaries}

\begin{definition}
An instance of the traveling repairman problem (TRP) is given by a set $V$ of $n+1$ points, one of them is the origin $r$, and symmetric integer distances $d_{ij}$ satisfying the triangle inequality. A solution is a permutation  $\Pi=(v_0,v_{1},\dots,v_n)$ of $V$, where $v_0=r$. The completion time $C(v)$ of a point $v=v_i$ is the distance from $r$ to $v$ on the path defined by $\Pi$: $C(v_i)=\sum_{j=1}^{i}d_{v_{i-1},v_{i}}$. The goal is to find a solution with minimum total completion time: $\sum_{v\in V} C(v)$.
\end{definition}

With loss of a $(1+\eps)$-factor in the approximation factor, we may assume that all distances are polynomially bounded.
Consequently, the length of the optimal tour is polynomially bounded. More precisely, (see also~\cite{Arora2003}) we may assume that all distances are in $\{0,1,\dots,B\}$ with $B=O(n^2/\eps)$.
We use the notation $\OO(.)$ when $\epsilon$ is assumed constant. For example, $\epsilon n=\OO(n)$ and $n^{1/\eps^2}=n^{\OO(1)}$.

It is convenient for the analysis to see a solution $\Pi$  as a path that is traversed with at most unit speed. 
That means, we assume there is a continues path of length $d_{v_{i-1},v_{i}}$ between  consecutive points $v_{i-1}$ and $v_{i}$ in $\Pi$.  The completion time $C(v)$ of input point $v$ is then defined as the time at which $v$ is visited for the first time. 

In~\cite{AroraKara2003}, the authors note that any solution $\Pi$ can be replaced by a concatenation of $\gamma=O(\log n/\eps)$ TSP-paths with only a $(1+\eps)$ factor increase in value. That means, the solution can be partitioned into $\gamma$ segments such that replacing each segment $S$ by a shortest path that visits the same points as $S$ and has the same start and endpoint as $S$, increases the value of the solution by at most a factor $(1+\eps)$. The proof follows easily by letting the number of points visited by the segments  decrease geometrically. Here, we prove the same lemma through the alternative approach of partitioning the timeline in intervals of geometrically increasing length. We shall not use  Lemma~\ref{lem:lognpaths} directly but will use a similar argument later when we partition the timeline in only a \emph{constant number} of intervals.

\begin{lemma}\label{lem:lognpaths}
With loss of a factor $1+\epsilon$ in the approximation, we may assume that $\OPT$ is a concatenation of~$O\left(\frac{\log n}{\eps}\right)$ TSP-paths.(\cite{AroraKara2003})
\end{lemma}
\begin{proof}
Consider time-points $1,(1+\eps),(1+\eps)^2,\dots, (1+\eps)^{\gamma}$, where $(1+\eps)^{\gamma}=O( n^3/\eps)$ is an upper bound on the length of the optimal tour. Then, ${\gamma}=O((\log n)/\eps)$. Now, replace the path between any two consecutive time-points by a TSP-path. The completion time of any point is increased by at most a factor $1+\eps$.  \qed
\end{proof}

By Lemma~\ref{lem:lognpaths}, it is enough to restrict to solutions composed of $\gamma=O\left((\log n)/\eps\right)$ TSP-paths.
In~\cite{AroraKara2003}, a solution composed of at most $\gamma$ TSP-paths is found by one dynamic program. Consequently, the $\log n$ shows up naturally in the exponent of the running time.
A simple example on the line shows that $\Omega(\log n)$ paths are needed for a PTAS: Let $n=2^k-1$ and place $2^{k-i}$ points at $x=(-2)^i$, for $i=1,\dots,k$. For this example, there is no constant approximate solution that is a concatenation of $o(\log n)$ TSP-paths.  
Hence, if we stick with the TSP-paths approach then the only way to improve on the running time is to have a better understanding of the dependency between the paths.
The key inside in our approach is that the TSP-paths can be clustered in groups of $K$ consecutive TSP-paths each, where $K$ is a  constant which depends on $\epsilon$ only, and such that there is only very limited dependency between the groups. That means the problem
on $O(\log n)$ TSP paths basically reduces to a problem on $K$ TSP paths.
Consequently, known TSP algorithms can be modified for these subproblems on $K$ TSP-paths. The dependency is limited in the sense that it is enough  to solve only a polynomial number of these subproblems. Then, dynamic programming is used to combine solutions for subproblems into one tour.

\subsection{Segmented TSP}
What we shall denote as the  segmented TSP is a generalization of the known $k$-TSP in which one needs to find, for a given TSP-instance and number $k\leqslant n$, a tour of minimal length that visits at least $k$ points. A $(2+\eps)$-approximation was given by Arora and Karakostas~\cite{AroraKarakostas:2000}. The problem can be solved exactly on a tree metric and a PTAS is known for Euclidean spaces of fixed dimension \cite{Arora1998}. For our PTAS, we need a more general problem that we denote by segmented TSP. It corresponds with the $k$-TSP problem for $K=1$.

\begin{definition}
 An instance of segmented TSP is given by a set $V$ of  $n+1$ points, one of them is the origin, and symmetric integer distances $d_{ij}$ satisfying the triangle inequality. Also  given are numbers $l_1\leqslant l_2\leqslant \dots \leqslant l_K$ and numbers $n_1\leqslant n_2\leqslant \dots \leqslant n_K$. A solution is a tour that starts and ends in the origin
 such that at least $n_h$ vertices are visited within the first $l_h$ distance for all $h\in \{1,\dots,K\}$ and such that the length of the tour is at most $l_K$.
  We say that an algorithm solves the problem if it always finds a solution if one exists. We say that an algorithm is an $\alpha$-approximation ($\alpha\geqslant 1$) if for any feasible instance it finds a tour that visits at least $n_h$ vertices within path-distance $\alpha l_h$ for all $h$ and such that its length is at most $\alpha l_K$.
\end{definition}

NB. One may also consider the segmented TSP without the restriction that the solution must end in the origin. This restriction is convenient for our purpose.

\begin{theorem}\label{th:main}
If, for any metric space, there is a polynomial time $\alpha$-approximation algorithm for the segmented TSP for every constant number of segments, then there is a polynomial time $\alpha(1+\eps)$-approximation algorithm for the Traveling Repairman Problem in the same metric space for every constant $\epsilon$.
\end{theorem}

In Section~\ref{sec:K-TSP},  we show that the segmented TSP with a constant number of segments can be solved exactly for weighted trees and show that there is a PTAS for the Euclidean plane.
\begin{corollary}
There exists a PTAS for the (unweighted) Traveling Repairman Problem in the Euclidean plane and for edge-weighted  trees.
\end{corollary}

The  following useful definition and lemma  apply to the segmented TSP  in general and are used in Section~\ref{sec:Reducing-to-Ktsp}.
\begin{definition}\label{def:CIj}
Let $I$ be a segmented TSP instance. The  $j$-th completion time of $I$ is denoted by $C_j^I$ and is defined as follows. The first $n_1$ completion times are $l_1$, the next $n_2-n_1$ completion times are $l_2$, and so on.
\end{definition}

Note that $C_j^I$ is an upper bound on the $j$-th completion time in any feasible solution for $I$. The following lemma is immediate.
\begin{lemma}\label{lem:CT le lambda CI}
Let $\TT$ be  a $\alpha$-approximate solution for segmented TSP instance $I$ and denote the $j$-th completion time in $\TT$ by  $C^{\TT}_j$. Then,
 $C^{\TT}_j\leqslant \alpha C_j^I\text{ for any $j$}$.
\end{lemma}

\section{Reducing TRP to segmented TSP}\label{sec:Reducing-to-Ktsp}
The reduction is done by the following steps. First, it is  shown that we may restrict to solutions that return in the origin at time points $t_i$, where $t_i/t_{i-1}=(1+\eps)^K$ for some large $K$ depending on $\eps$ only. For this, we use a simple probabilistic  argument. The part of the tour between time points $t_{i}$ and $t_{i+1}$ is called the $i$-th. \emph{subtour}.
Each of these $\Gamma=\OO(\log n)$  subtours can be partitioned into $K$ subpaths where the ratio of end time and start time of each path is $1+\eps$. We call these subpaths \emph{segments}. In the optimization, we may approximate the completion time of a point by the endpoint of the segment that it is on.
If we would know for each subtour the points to be visited, then an approximate solution can easily be computed given a segmented TSP algorithm.
Clearly, we cannot afford to guess these subsets. However, as we show in this section, for large enough $K$, we can afford to revisit in subtour $i$, all points that were visited in the preceding subtours.
Consequently, in the dynamic programming there is no need to keep track of \emph{subsets} of points and we only need to enumerate over the \emph{number} of points visited.
This requires only a polynomial number of segmented TSP instances to solve.
Remarkably, the number $\Gamma$ of subtours is not dominating the running time, which is only polynomial in $\Gamma$.

\subsection{Restricting the solution space}\label{sec:restricting solution}

Assume $0<\eps\leqslant 1$ and let $K$ be an integer depending on $\epsilon$ only.  To simplify notation, we write $\delta=1+\eps$. Choose $h_0$ uniformly at random in $\{0,1,\dots,K\!-\!1\}$ and let
\begin{equation}\label{eq:h0}
A_i=\delta^{(i-1)K+h_0}, \text{ for }i\geqslant 0.
\end{equation}
Consider an optimal solution, $\OPT$, and let $L$ be its length, i.e., the largest completion time. 
(In general, we denote by $\OPT$ the solution itself as well as its value.)  
For $i\geqslant 1$, let $\OPT_i$ be the solution restricted to the first length $A_i$. (Note in particular that $\OPT_1$ has length $\delta^{h_{0}}$.) Let $\Gamma$ be the smallest integer such that  $A_{\Gamma}\geqslant L$. Hence, we may assume\footnote{More precisely, we have  $A_{\Gamma}\geqslant L=O(n^3/\eps)$, where $A_{\Gamma}\ge(1+\eps)^{({\Gamma}-1)K}$.
Hence, $\Gamma K=O(\log_{1+\eps} (n^3/\eps))$ $\Rightarrow$  $\Gamma=O(\log n /(\eps K)=O(\eps\log n)$ if we take $K=\Theta(1/\eps^2)$.}
 $\Gamma=\OO(\log n)$.
 The modified solution $\OPT'$ is defined as follows:\bigskip

\noindent $\OPT'$: For $i=1$ to $\Gamma$, start $\OPT_i$ at time $t_i:=3A_{i-1}$ and return to the origin.\bigskip

(The constant 3 above may be replaced by any constant strictly larger than 2 for the proof to work.)
 Let $v$ be an arbitrary point of the instance and let $C(v)$ and $C'(v)$ be its completion time in, respectively, $\OPT$ and $\OPT'$, where the completion time is the first moment that the point is visited. Let $\expected[\OPT']$ be the expected value of $\OPT'$ over the random choice of $h_0$.
\begin{lemma}\label{lem:decomposition}
For large enough $K=O(1/\eps^2)$, it holds that
$\expected[C'(v)]\leqslant (1+\eps)C(v)$ for any input point $v$. Hence,
$\expected[\OPT']\leqslant (1+\eps)\OPT$.

\end{lemma}
\begin{proof}
Feasibility holds if it is possible to return to the origin after each $\OPT_i$ before beginning the next path $\OPT_{i+1}$ at time $t_{i+1}$. This is clearly true if $t_{i}+2A_i\leqslant t_{i+1}$ for all $i$. Since $t_i=3A_{i-1}$, this is equivalent with \[
3A_{i-1}+2A_i\leqslant 3A_{i}\ \Leftrightarrow A_i/A_{i-1}\geqslant 3 \Leftrightarrow \delta^K\geqslant 3.\]
Hence, for feasibility it is enough to take $K=O(1/\log \delta )=O(1/\eps)$.
Now, let us compute the expected value of $\OPT'$.  Consider an arbitrary point $v$ of the instance and let $i'$ be the smallest index such that $A_{i'}\geqslant C(v)$, that means, point $v$ is visited in $\OPT'$ for the first time by path $\OPT_{i'}$. Let $\delta^{q-1}<C(v)\leqslant \delta^{q}$, for some integer $q\geqslant 0$. (Note that $q\geqslant 0$ since the minimum distance and hence the minimum completion time is at least 1.) Then the expected value of $A_{i'}$ is
\begin{equation*}\label{eq:expected length}
\expected[A_{i'}]=\frac{1}{K}\sum_{h=0}^{K-1} \delta^{q+h}=\frac{\delta^{q+K}-\delta^q}{K(\delta-1)}<\frac{\delta^{q+K}}{K(\delta-1)}.
\end{equation*}
Remember that $t_{i'}=3A_{i'-1}=3\delta^{-K}A_{i'}$ and note that $t_{i'}=C'(v)-C(v)$. Hence,
\begin{eqnarray*}\label{eq:expected Cj}
\expected[C(v')]-C(v)&=&\expected[t_{i'}]=\frac{3}{\delta^K}\expected[A_{i'}]
<\frac{3}{\delta^K}\frac{\delta^{q+K}}{(\delta-1)K} \\
&=&\frac{3\delta^q}{(\delta-1)K}<\frac{3\delta}{(\delta-1)K}C(v).
\end{eqnarray*}
It follows that
$\expected[C'(v)]\leqslant (1+\eps)C(v)$, for any point $v$ if $\frac{3\delta}{(\delta-1)K}\leqslant \eps$, i.e., if
\[K\geqslant \frac{3\delta}{ \epsilon(\delta-1)}=\frac{3(1+\eps)}{\eps^2}=O\left(\frac{1}{\eps^2}\right).\]
\qed
\end{proof}

Since $\expected[\OPT']\leqslant (1+\eps)\OPT$ there must be some $h_0$ for which the corresponding deterministic solution $\OPT'$ satisfies  $\OPT'\leqslant (1+\eps)\OPT$. From now on  we consider $\OPT'$ to be this deterministic solution.
For $j=1,\dots,n$, let $D_j$ be the $j$-th completion time of $\OPT'$ (where we only consider the first appearance of each point). Equivalently, we can define $D_{j}$ as the completion time of the $j$-th point on the first  subtour that visits at least $j$ points.
The properties of solution $\OPT'$ are listed in the next lemma.
\begin{lemma}\label{lem:properties}
Solution $\OPT'$ has the properties:	
\begin{enumerate}
\item[(i)] It is in the origin at time $t_i=3\delta^{(i-2)K+h_0}$ for  all $i=1,2,\dots, \Gamma$, where $K=O\left(\frac{1}{\eps^2}\right)$, $\Gamma=O(\eps\log n)$ and $h_0$ is some fixed number in $\{0,1,\dots, K-1\}$. We call the tour between $t_i$ and $t_{i+1}$ the $i$-th subtour.
\item[(ii)] The number of points on the $i$ th. subtour is non-decreasing in $i$.
\item[(iii)] For $j=1,\dots,n$, let $D_j$ be the completion time of the $j$-th point on the first  subtour that visits at least $j$ points. Then $\sum_{j=1}^n D_j\leqslant (1+\eps)\OPT$.
\end{enumerate}
\end{lemma}

Now consider any tour that satisfies properties (i) and (ii) and let $D_j$ be defined as in (iii). Then, clearly the $j$-th completion time is no more than $D_j$. Hence,  the lemma shows that we may restrict to solutions which have properties $(i)$ and $(ii)$ and among those tours minimize  $\sum_{j=1}^n D_j$ as defined in $(iii)$.  We shall prove that minimizing $\sum_{j=1}^n D_j$ can be done easily by dynamic programming if we have an algorithm for the following subproblem.

\subsection{The subproblem}
\begin{definition}An instance of the subproblem is given by $i\in\{1,\dots,\Gamma\}$  and numbers $m'\leqslant m''\in \{0,1,\dots,n\}$. A solution is a tour that starts at the origin at time $t_i$ and returns before time $t_{i+1}$ and visits exactly  $m''$ points. The value of a solution is the sum of completion times of points  $m'+1,\dots,m''$ on this tour (which is zero if $m'=m''$). The objective is to find a solution with minimum value.  Note that an instance $i,m',m''$ may not have a feasible solution. For any feasible instance, let $\SUB_i(m',m'')$ be its optimal value.\end{definition}

Let $m_i$ be the number of points visited by the partial solution $\OPT_i$. Then clearly, 
\[
\OPT'\geqslant \sum_{i=1}^{\Gamma}\SUB_i(m_{i-1},m_i).
\]

\begin{definition}
An $(\alpha,\beta)$ approximation algorithm for the subproblem is an algorithm that finds for any \emph{feasible} instance  $(i,m',m'')$ a tour that starts in the origin at time $\alpha t_i$ and ends in the origin before  time $\alpha t_{i+1}$, visits exactly $m''$ points, and for which the total completion time of the points $m'+1,\dots,m''$ is at most  $\alpha\beta\SUB_i(m',m'')$.
\end{definition}

Assume we have an  $(\alpha,\beta)$-approximation algorithm $\ALG$ for the subproblem. Let $\Alg_i(m',m'')$ be the value returned by the algorithm for instance $(i,m',m'')$ and let it be infinite if no solution was found. 
For any sequence of integers $0\le\hat{m}_1\leqslant \dots \le\hat{m}_{\Gamma}=n$ we get a tour of total completion time 
\begin{equation}\label{eq:total_ALG}
\sum_{i=1}^{\Gamma}\alg_i(\hat{m}_{i-1},\hat{m}_i)\leqslant \alpha\beta \sum_{i=1}^{\Gamma}\SUB_i(\hat{m}_{i-1},\hat{m}_i)
\end{equation}
by concatenating the tours $\alg_i(\hat{m}_{i-1},\hat{m}_i)$.
Minimizing the left side of~\eqref{eq:total_ALG} over all values $0\le\hat{m}_1\leqslant \dots \le\hat{m}_{\Gamma}=n$ is easy since they form a non-decreasing sequence. To be precise, let $\Alg_1(m'')=\alg_1(0,m'')$ for all $m''\leqslant n$ and for $k=2,\dots,\Gamma$, let 
\[\Alg_{k}(m'')=\min_{m'\leqslant m''}\Alg_{k-1}(m')+\alg_k(m',m'').\] 
Then, the minimum is given by $\Alg_{\Gamma}(n)$.  Let the values $\hat{m}_i$ minimize the left side of~\eqref{eq:total_ALG} and let $m_i$ be the number of points visited by the partial solution $\OPT_i$. Then, we find a solution of total completion time at most  
\begin{eqnarray*}
\alpha\beta \sum_{i=1}^{\Gamma}\SUB_i(\hat{m}_{i-1},\hat{m}_i)&\le& \alpha\beta \sum_{i=1}^{\Gamma}\SUB_i(m_{i-1},m_i)\\
&\leqslant & \alpha\beta\OPT'\le\alpha\beta(1+\eps)\OPT.
\end{eqnarray*}
The number of subproblems is $O(\Gamma n^2)$ and given all approximate values, the dynamic programming takes $O(\Gamma n^2)$ time. Further, the number of choices for $h_0$ is $K$ (See Equation~\ref{eq:h0}). Hence, it takes only $O(K\Gamma n^2)=\OO(n^2\log n)$ calls to the approximation algorithm for the subproblem to get an $\alpha\beta(1+\eps)$-approximation for the \TRP.\bigskip

\paragraph{Approximating the subproblem.} We show how to obtain an $(\alpha,1+\eps)$-approximation for the subproblem if we have an $\alpha$-approximation algorithm for the segmented TSP.
Let $(i,m',m'')$ be a feasible instance of the subproblem.
For $h=0\dots K$, define \emph{time-point}
\begin{equation}\label{eq:t_ij}
t_i^{(h)}=(1+\eps)^h t_i,\   (\text{Hence, } t_i^{(K)}=t_{i+1}^{(0)}=t_{i+1}.)
\end{equation}
Recall the definition of the  segmented TSP problem.  A polynomial number of segmented TSP instances is solved (approximately). Let $l_h=t_i^{(h)}-t_i$, $h=1,\dots,K$ and hence, these are fixed given the index $i$. The numbers $n_h$ take all possible integer values for which $n_1\leqslant n_2\leqslant \dots \leqslant n_K=m''$.  This gives $O(n^K)$  instances. Solve all these instances by some $\alpha$-approximate  segmented TSP algorithm and determine the solution with smallest total completion time of the points $m'+1,\dots,m''$. Let $T$ be this solution and let $T'$ be the solution $T$ started at time $\alpha t_i$. We show that $T'$ is an $(\alpha,1+\eps)$-approximation for the subproblem  $(i,m',m'')$. 

The length of $T$ is at most $\alpha(t_{i+1}-t_i)$. Hence, $T'$ completes before time $\alpha t_{i+1}$. Also, it visits exactly $m''$ points. Now let $\alg_i({m}_{i-1},{m}_i)$ be the value of $T'$ for subproblem $(i,m',m'')$.
Consider an optimal solution $\Pi$ for subproblem $(i,m',m'')$ and let $I$ be the segmented TSP instance given by the numbers $n_h$, where $n_h$ is the number of points visited by $\Pi$ until time $t_i^{(h)}$. 
Let $C_j^{\Pi}$ be the $j$-th completion time in $\Pi$ and let $C_j^{I}$ be the $j$-th completion time of $I_{\Pi}$  as defined in Definition~\ref{def:CIj}. Then, 
\[(1+\eps)C_j^{\Pi}\geqslant (t_i+C_j^{I}).\]
Instance $I$ is among the enumerated instance. Hence, using Lemma~\ref{lem:CT le lambda CI},
\begin{eqnarray*}
\alg_i({m}_{i-1},{m}_i)&\leqslant &(m''-m')\alpha t_i+ \alpha \sum_{j=m'+1}^{m''}C_j^I\\
&= &\alpha\sum_{j=m'+1}^{m''}(t_i+C_j^I)\\
&\leqslant &\alpha \sum_{j=m'+1}^{m''} (1+\eps)C_j^{\Pi}\\
&= & (1+\eps)\alpha\SUB_i(m_{i-1},m_i).
\end{eqnarray*}

\paragraph{Running time.}  
For each subproblem,  $O(n^K)$  segmented TSP instances are solved and we simply store the best one. There are $\OO(n^2\log n)$ instances for the subproblem and the dynamic program runs in $\OO(n^2\log n)$ time. 
Hence, the total running time is $n^{O(K)}=n^{O(1/\eps^2)}$ multiplied by the running time of the $\alpha$-approximation algorithm for the segmented TSP.

\section{Approximating the  segmented TSP}\label{sec:K-TSP}
By Theorem~\ref{th:main}, any $\alpha$-approximation algorithm for segmented TSP implies a $(1+\eps)\alpha$-approximation algorithm for the \TRP\ in the same metric space. Here, we consider the approximability of segmented TSP in different metric spaces.  Remember the definition of an $\alpha$-approximation algorithm for the segmented TSP problem: It finds a solution such that $n_i$ points are visited before time $\alpha l_i$, where $n_i$ and $l_i$ are given. (For ease of notation we use an index $i$ instead of $h$ as used in the previous section.) When, we assume that the number of segments $K$ is constant, then we may guess the number of points visited on each of the segments.
More precisely, we denote by segment $i$, the path that runs between distance $l_{i-1}$ (excluded) and $l_i$ (included). We guess the numbers $\mu_i$ of points visited on segment $i$, where  $n_i=\sum_{j=1}^i \mu_j$ for all $i$. The number of choices is only $O(n^{{K}})$, which is polynomial if $K$ is a constant. Further, we assume that all $l_i$ are integer and denote $\lambda_i=l_{i}-l_{i-1}$. Hence, from now, we assume that the segmented TSP instance is given by numbers $\lambda_i$ and $\mu_i$ ($i=1,\dots,K$) and we need to visit \emph{exactly} $\mu_i$ points on the $i$-th segment.

N.B. By `guessing' we mean enumerating over all possible values and we say that we are \emph{able to guess} a certain value if the number of possible values is polynomialy bounded.

\subsection{Edge-weighted tree}
The metric space is given by a tree $T$ with non-negative integer weights on the edges. The distance between any two points $u,v$ is the length of the unique path between $u$ and $v$ on $T$. The TSP is trivial on trees since a tour is optimal if and only if it is a depth-first search on $T$. Also, the $k$-TSP can easily be solved by dynamic programming: For each vertex $v$ and number $j\leqslant k$, store the length $l(v,i)$ of the shortest tour in the subtree rooted at $v$ which visits exactly $j$ vertices. The value is easily computed from the table of values of the children of $v$.

The generalization to segmented TSP is straightforward.
First, turn the tree into a rooted binary tree such that only leaves need to be visited. This can be done with only a constant factor increase in the number of points by adding edges of length zero. For each node $v$ unequal to the root we define a \emph{vector of crossing information} as follows.
The edge above $v$ is traversed at most $2K$ times.
This gives at most $K$ subtours in the tree rooted at $v$ which start and end at $v$. For each of these we guess the start time and end time and we guess the number of points that each of the $K$ segments have on this subtour. For all possible vectors we only store if this is feasible or not. A vector is feasible if it can be obtained from feasible vectors of its two children. For any leaf, a vector is feasible if there is exactly one subtour and the  start time equals its end time and it contains exactly one vertex (namely $v$). Note that the time of visit determines the segment that $v$ is on. For the root we only consider the case of one  subtour starting at time $0$ and ending at time $l_K=\sum_{i=1}^K \mu_i$ and for which segment $i$ contains exactly $\mu_i$ points. The running time is $n^{O(K)}$.

\subsection{Euclidean plane}\label{sec:Euclidean}
We show that for any feasible segmented TSP instance we can find a $(1+\eps)$-approximate solution in time $n^{O(K)}$. That means, the solution is a concatenation of $K$ paths, where the $i$-th path has length at most $(1+\eps)\lambda_i$ and visits exactly $\mu_i$ points. It is important to note that the $\epsilon$ used in this section has nothing to do with the $\epsilon$ of Section~\ref{sec:Reducing-to-Ktsp}. That means, in this section, $K$ is an arbitrary integer constant.

Arora and Karakostas~\cite{AroraKara2003} give a quasi polynomial time approximation scheme for the \TRP\ in the Euclidean plane. (See also~\cite{Arora2003}.)
The algorithm in~\cite{AroraKara2003} is based on the refined TSP-PTAS~\cite{Arora1998}, which is more efficient than the simpler version that was published earlier~\cite{Arora1996}. In the latter paper, it was shown that here is a $(1+\eps)$-approximate TSP tour that crosses the boundary of each square in the quadtree only $O(\log n /\eps)$ times. In the refined PTAS, it was proven that  $O(1/\eps)$ crossings satisfy too. In combination with Lemma~\ref{lem:lognpaths} this led the authors of~\cite{AroraKara2003} to a TRP algorithm with $n^{O(\log n/\eps^2)}$ running time. The proof contains many details but intuitively it does follow easily from the next three observations: (i) all lengths are polynomially bounded, (ii) the solution is composed of $O(\log n/\eps)$ TSP-paths, and (iii) there are only $O(1/\eps)$ crossings per square per TSP-path. Hence, for a given square we can afford to guess for each crossing basically all information that we want and still end up with quasi-polynomial running time.
In the segmented TSP problem, the solution is composed of only $K$ TSP paths. Hence, for constant $K$ we should expect a better running time. A minor issue is that we have a restriction on the length of each of the $K$ segments. This is easily solved by using Markov's inequality, as we show below in the discussion of the structure theorem. 
Our PTAS for Euclidean segmented TSP applies even if we adopt the simpler TSP PTAS~\cite{Arora1996} that allows $O(\log n/\eps)$ crossings of the dissection squares.
The TSP-PTAS~\cite{Arora1996,Arora1998,Arora2003} contains numerous details. Here we only address those that are of interest for our modification and refer to the survey~\cite{Arora2003} for omitted details.

\paragraph{Structure theorem} The rounding of the instance and the construction of the quadtree and portals remains basically the same: Take the smallest  bounding box and define a grid of polynomial dimension. Move input points to the middle of grid cells. Then, place an enclosing box of double side length at random on top of it. Next, make the dissection tree. The depth is $\OO(\log n)$. We let the number of portals for each dissection square be $O(\log n/\eps)$.
By scaling distances, we may assume that for each grid cell and segment $i$, the part of the segment that lies inside the cell has integer length.

For the Euclidean  TSP problem it is known~\cite{Arora2003} that there is a tour $\Pi$ that crosses the boundary of each dissection square only in portals, and at most twice in each portal, and for which the expected length is at most $(1+\eps)$ times optimal. The same is true for the Traveling Salesman Path problem~\cite{Arora2003}. The expectation is over the random shift of the enclosing  box. More precisely, for any path of length $S$ in the bounding box, the  expected length of the detour that is needed to make it portal respecting is $\eps S$.
Now consider a feasible segmented TSP instance given by numbers $\lambda_i$ and $\mu_i$. It follows directly that there is solution $\TT$ that crosses only at portals and each portal at most $2K$ times such that each segment $i$ visits $\mu_i$ points and has length $L_i\geqslant \lambda_i$ and such that $\expected[L_i-\lambda_i]\leqslant \eps\lambda_i$. Again, the expectation is over the random shift of the box. Note that this is not enough for our purpose since we want each of the $K$ differences  to be at most $\eps\lambda_i$ simultaneously. Since $K$ is constant, this is easily solved by Markov's inequality: $\Pr[L_i-\lambda_i\geqslant 2K\eps\lambda_i]\leqslant 1/(2K)$ for each $i$. Then, by the union bound, $\Pr[L_i-\lambda_i\geqslant 2K\eps\lambda_i \text{  for at least one }i]\leqslant 1/2$. Hence, in stead of an expected $(1+\eps)$-approximate solution we get a $(1+2K\eps)$-approximate solution with probability at least $1/2$. The additional factor $2K$ is no issue since $K$ is a constant.
(Again, remember that $K$ is an absolute constant independent of $\epsilon$ in this section.)

\paragraph{Dynamic Programming} Note that in the dynamic programming we do not solve an optimization problem but only search for a feasible solution. An instance $I$ of a subproblem in the DP is given by:
\begin{enumerate}
\item[(1)] A dissection square $S$.
\item[(2)] For each segment $i$, the number of points and the length of segment $i$ inside $S$.
\item[(3)] For each portal of $S$ and all segments $i$, the  number of times segment $i$ crosses it (0,1,or 2) and in which direction (in or out).
\item[(4)] For each segment $i$, the first and last crossing with $S$ are specified.
\item[(5)] A pairing of the crossings with $S$.
\end{enumerate}
Note that we only guess the length and number of points for each segment and not for each crossing as was done in~\cite{AroraKara2003}. Hence, we can afford $\OO(\log n)$ crossings.
Clearly, the number of choices for items (1)--(4) is $n^{O(K/\eps)}$. The pairing of the crossings can be done almost independently for each segment since we know for each crossing the segment it belongs to and we know the first and last crossing of each segment. Hence, the number of pairings is bounded by $2^{O(\log n/\eps)K}=n^{O(K/\eps)}$.

First, consider the base case. By the rounding step, all points coincide and are in the middle of the cell. Clearly, it would be optimal to serve all these  by the same segment. However, we assumed the number of points on each segment $i$ to be given by $\mu_i$. Hence, we should allow the midpoint to be visited by multiple segments. Clearly, each segment needs to  cross the midpoint at most once. Feasibility can be checked in  $O(Km)$ time, where $m=O(\log n/\eps)$ is the number of portals per square. For the smallest dissection square containing the root vertex we have the additional restriction that segment 1 starts in the root and segment $K$ ends in the root.

Consider an arbitrary (non-base) instance $I$ given by (1)--(5). We check if there is a feasible instance for each of its children which together are consistent with instance $I$. That means, the number of points and lengths should add up to the right value and all crossing and pairings should be consistent. Further, one needs to exclude combinations that form subtours. For each instance $I$ there are $n^{O(K/\eps)}$ combinations of instances for its four children to check. The time for checking a single combination is only linear in the number of portals. 
For the largest square (the enclosing box)  we only need to verify one instance: the $(1+\eps)$-approximation, i.e., there are no crossings and segment $i$ has length $(1+\eps)\lambda_i$ and contains exactly $\mu_i$ points.

\section{Generalizations and variants}

The approximation schemes for TRP in $\RR^2$ and weighted trees apply as well if \emph{release dates} are added. The transformation from $\OPT$ to $\OPT'$ works still fine in that case since the solution is only moved forward in time. Hence, $\OPT'$ is feasible and the total completion time is increased by at most a factor $1+\eps$.
In the reduction to segmented TSP, we need to consider segmented TSP instances with release dates. By rounding release times (by at most a factor $1+\eps$)  we may assume that points are released only at the start times of the $K$ segments. Equivalently,  we may assume that we have sets $S_1\subseteq S_2\subseteq \dots \subseteq S_K$ of points such that the $j$-th segment can only visit points from $S_j$. In the dynamic programs, except for the base case, we do not consider  which points are visited but only store the number for each of the segments. The base case can still be efficiently solved since the number of segments is constant.

The PTAS applies as well if our objective function is a linear combination of total completion time, $\sum_{j} C_j$, and the length of the path. That means, the problem is to find a path, starting in the origin, that minimizes $\alpha\sum_{j} C_j +\beta \max_{j} C_{j}$ for some $\alpha,\beta\geqslant 0$. To see this, define the path $\OPT'$ in exactly the same way. For any input point $p$ we have  $\expected[C'(p)]\leqslant (1+\eps)C(p)$. In particular, this applies to the last point on the path. Hence, $\expected[\max_{j} C'_{j}]\leqslant (1+\eps)\max_{j} C_{j}$, where $C_{j}$ ($C_{j'}$) is the $j$-th completion time in $\OPT$ ($\OPT'$). In total we get that
\[\expected[\alpha\sum_{j} C'_j +\beta \max_{j} C'_{j}]\leqslant (1+\eps)(\alpha\sum_{j} C_j +\beta \max_{j} C_{j}).\]
In the algorithm we guess $h_0$ for which the inequality above holds without expectation. Also, we guess the corresponding length $L'$ of the tour $\OPT'$. Then we apply the same DP but we restrict to tours of length at most $(1+\eps)L'$.

In the \emph{Randomized Search Ratio} problem one has to find a (random) path starting from the root $r$ and visiting all points and the goal is to minimize $\max_v \expected[C(v)]/d(r,v)$, where $d(r,v)$ is the distance from $r$ to $v$ .  In~\cite{AroraKara2003}, the authors mention that E. Tardos observed the following: If the minimum latency problem has a
PTAS for a certain class of metrics, then the randomized search ratio problem has an
approximation scheme for that same class of metrics. Thus, our PTAS implies a PTAS for the randomized search ratio for trees and the Euclidean plane.

The PTAS also applies to the \emph{The $k$-repairman problem} in which one needs to find $k$ repairman paths that together visit all points.  The transition from $\OPT$ to $\OPT'$ is the same: All repairman are in the origin at the same time. In the segmented TSP we need to find $k$ segmented TSP-paths simultaneously. For constant $k$, there is only a polynomial increase in the running time.

\subsection{Open problems}
\paragraph{Weighted completion times}
The generalization to weighted completion times is straightforward if weights are polynomially bounded. However, for general weights it is not clear how to adjust the approximation scheme. 

\paragraph{Metric embedding on a line}
Another interesting problem that is closely related is that of finding a metric embedding on a line such that the average distortion is minimized~\cite{DhamdhereGR2006}. One can show that the $k$-TRP with $k=2$ is a special case of this metric embedding problem. 
The authors of~\cite{DhamdhereGR2006} use ideas of the QPTAS for the traveling repairman problem to obtain a QPTAS for the average distortion problem. It is not clear whether our ideas can be used to obtain a PTAS for metric line-embedding as well.

\paragraph{Weighted planar graphs}
In~\cite{AroraKara2003} the authors remark that their quasi-PTAS for Euclidean TRP carries over directly to weighted planar graphs by using the PTAS for TSP on weighted planar graphs by Arora et al.~\cite{AGKKW}. This claim turned out to be incorrect (Karakostas, personal communication, 2014). For the TSP PTAS the separator is a Jordan curves that divides the graph into an exterior and interior part. The number of portals is $m = O(\log n/\eps^2)$ and each portal is crossed at most twice. Hence, the situation here is similar to the Euclidean case. However, in the planar case the graph is first reduced by contracting some of the edges. Uncontracting the edges increases the length of the tour by at most a factor $1+\epsilon$. This is fine for the TSP but is problemetic for the TRP: Uncontracting edges early in the TRP-path may cause a large increase in the total completion time. This issue was overlooked in~\cite{AroraKara2003} and in earlier versions of this report~\cite{Sitters:14}. Getting a PTAS or even a quasi-PTAS for the TRP in weighted planar graphs remains an open problem. Moreover, this also holds for the $k$-TSP problem in weighted planar graphs.

\section{Single machine scheduling under precedence constraints}

The reduction used for the TRP applies to almost any problem of minimizing the total (weighted) completion, assuming that weights are polynomially bounded. Of course, this doesn't mean that it is always useful since the subproblem may be harder to approximate than the original. First, we give a rough sketch how to apply it to the simple scheduling problem $1|r_j|\sum C_j$ and then give a detailed proof for the more challenging problem of scheduling under precedence constraints. A PTAS for the first was given by Afrati et al~\cite{Afrati:99}. \bigskip

\textbf{Example}:$1|r_j|\sum C_j$. We have a single machine and $n$ jobs with processing times $p_j$ and release times $r_j$ for $j=1,\dots,n$. The objective is to find a schedule that minimizes the total completion time $\sum_j C_j$, where $C_j$ is the completion time of job $j$. Now, the subproblem is defined on an interval from $t_i$ to $(1+\eps)^Kt_i$ for some $i$, and where $K$ depends on $\eps$ only. For given $n'\leqslant n''$, the problem is to find a feasible schedule on a subset of the jobs that minimizes $\sum_{j=n'+1}^{n''} C_j$. Now partition the interval in $K$ subintervals as before where the ratio of start and end time of a subinterval is $1+\eps$. Hence, we may assume that jobs are released only at the beginning of subintervals. Say that a job is large if its processing time is more than $\eps$ times the length of the smallest subinterval (which is the first). Then, the number of large jobs in the optimal solution to the subproblem is bounded by a constant and we guess all of them. The small jobs can be  added greedily such that each subinterval is overpacked by at most $\epsilon$ times its length. \qed\bigskip

One of the most intriguing scheduling problems is that of minimizing total weighted completion times on a single machine under precedence constraints. ($1|prec|\sum_{j} w_{j}C_{j}$, in the notation by Graham et al~\cite{GrahamLLR1979}.) The problem is known to be \NP-hard~\cite{Lawler1978,LenstraRK1978} and several 2-approximation algorithms are known. The paper by Amb\"uhl et al.~\cite{AmbuhlMMS2011} gives a recent overview on the status of this problem. Exact polynomial time algorithms are known for some special cases, e.g., for series parallel possets~\cite{Lawler1978}. Surprisingly, interval ordered precedence constraints are not one of these. Woeginger~\cite{Woeginger2003} gave a 1.62-approximation algorithm and a $3/2$-approximation was given by Amb\"uhl et al~\cite{AmbuhlMMS2011}. The same paper shows that scheduling interval orders is in fact \NP-hard. Here, we give a polynomial time approximation scheme for interval ordered precedence constraints.

An instance of the  scheduling problem is given by $n$ jobs to be processed on a single machine that can process at most one job at a time. Each job $j$ has a nonnegative integer processing time $p_j$ and weight $w_j$. A partial order on the jobs defines the precedence constraints between jobs. That means, if $j_1\prec j_2$, then job $j_1$ must be completed before $j_2$ can start. The goal is to find a non-preemptive schedule that minimizes $\sum_{j=1}^n w_jC_j$, where $C_j$ is the completion time of job $j$.

\begin{definition}
A partial order on a set $J$ is an interval order if there is a function that assigns to each $j\in J$ a closed interval $[l_{j},r_{j}]$ such that $j_1\prec j_2$ if and only if $r_{j_1}< l_{j_2}$.
It is easy to see that for any interval order there is a corresponding set of intervals for which all $2|J|$ endpoints are different.
\end{definition}

A theorem by Woeginger~\cite{Woeginger2003} states that for general precedence constraints, we may restrict our approximation analysis to the case $1\leqslant p_{j}\leqslant n^{2}$ and $1\leqslant w_{j}\leqslant n^{2}$, where $n$ is the number of jobs. In fact, this theorem  can be applied to the special case of interval orders since its proof only reverses the precedence constraints, and since the reverse of an interval order is again an interval order (see \cite{Woeginger2003}).

\subsection{Reducing the problem}\label{sec:redcution-for-scheduling}
The reduction is almost the same as what we did for the TRP problem. Let $K=O(1/\eps^{2})$ and choose $h_0$ uniformly at random from $\{0,1,\dots,K-1\}$. The numbers $A_{i}$ are as before.
Consider an optimal solution, $\OPT$, and let $\OPT_i$ be the solution restricted to the jobs that complete not later  than $A_i$, for $i=1,2,\dots $.
The  schedule $\OPT'$ is defined by simply concatenating all the solutions $\OPT_{i}$. Note that jobs appear multiple times since any job that appears in $\OPT_{i}$ appears as well in $\OPT_{i'}$ for all $i'\geqslant i$. In general, we allow jobs to appear more than once and call these \emph{pseudo schedules}. The completion times and precedence constraints apply only to the first appearance of each job.  

As before, denote by $t_{i}$ the time at which $\OPT_{i}$ starts in $\OPT'$. The solution $\OPT'$ is well-defined if $t_{i}\geqslant t_{i-1}+A_{i-1}$ for all $i$. Let $t_{i}=c A_{i-1}$, then 
$c A_{i-1}\geqslant c A_{i-2}+A_{i-1}$  holds if $c/(c-1)\geqslant A_{i-1}/{A_{i-2}}$. For any constant $c>1$ we can choose $K$ such that $c/(c-1)\geqslant A_{i-1}/A_{i-2}=\delta^K$.
For simplicity, let us just take $t_{i}=3A_{i-1}$ as before. This creates unnecessary idle time but at least we can blindly copy the analysis of the TRP. Let $C_{j}$ ($C'_{j}$) be completion time of job $j$ in $\OPT$ ($\OPT'$). Then, following the proof of Lemma~\ref{lem:decomposition}, we have for any job $j$ that
\[\expected[C_{j}']\leqslant (1+\eps)C_{j},\]
where the expectation is over the random choice of $h_{0}$. Taking the weighted sum we have
\[\expected[\OPT']=\expected[\sum_{j}w_{j}C_{j}']\leqslant (1+\eps)\sum_{j}w_{j}C_{j}=(1+\eps)\OPT.\]
From now assume that $h_{0}$ is chosen such that the inequality holds without expectation: $\OPT'\leqslant (1+\eps)\OPT$.

We call the schedule between two consecutive time points $t_{i}$ a \emph{subschedule}. Note that in $\OPT'$, each subschedule is a feasible schedule on its own. The total weight of jobs in the $i$-th subschedule is the weight completed by $\OPT_i$ and hence, is non-decreasing in $i$. 
Let $W=\sum_{j}w_{j}$. Then, $W\leqslant n^{3}$, since $w_j\leqslant n^2$ for all $j$.  For any $w\in \{1,2,\dots,W\}$, let $D^{w}$ be the first moment at which $\OPT'$ completes a total weight of at least $w$ (where for any job we only count the weight of its first appearance and the weight is only counted when the job completes.) Equivalently, we may define $D^{w}$ as the first moment at which some subschedule completes a total weight of at least $w$.
Then,
\begin{equation}\label{eq:OPT'D^w}
\OPT'=\sum_{j}w_{j}C'_{j}=\sum\limits_{w=1}^{W} D^{w}.
\end{equation}
The properties of the pseudo schedule $\OPT'$ are listed in the next lemma.
\begin{lemma}\label{lem:propertiesSchedule}
Solution $\OPT'$ has the following properties:	
\begin{enumerate}
\item[(i)] 
No job is processed at time $t_i$ and the subschedule between time points $t_i$ and $t_{i+1}$ is a feasible schedule on itself. Here, $t_i=3(1+\eps)^{(i-2)K+h_0}$ for  all $i=1,2,\dots, \Gamma$, where $K=O\left(\frac{1}{\eps^2}\right)$, $\Gamma=O(\eps\log n)$ and $h_0$ is some fixed number in $\{0,1,\dots, K-1\}$.
\item[(ii)] The total weight of jobs scheduled  in the $i$-th subschedule is non-decreasing in $i$.
\item[(iii)]
For any $w\in \{1,2,\dots,W\}$, define $D^{w}$ as the first moment at which some subschedule completes a total weight of at least $w$. Then,
$\sum\limits_{w=1}^{W} D^{w}\leqslant (1+\eps)\OPT$.
\end{enumerate}
\end{lemma}
Now consider any pseudo schedule that satisfies (i) and (ii) and let $D^{w}$ be as defined in $(iii)$ and let $C^{w}$ be the moment that the schedule completes a total weight of at least $w$. Then (using~\ref{eq:OPT'D^w})
\[\sum_j w_jC_j=\sum\limits_{w=1}^{W} C^{w}\leqslant \sum\limits_{w=1}^{W} D^{w}.\]
(Equality holds for $\OPT'$.) Hence, we may restrict to pseudo schedules which have properties $(i)$ and $(ii)$ and among those, minimize  $\sum\limits_{w=1}^{W} D^{w}$ as defined in $(iii)$.  This can be done approximately by dynamic programming as before if we have an approximation algorithm for the following subproblem on subschedules.\bigskip

\paragraph{Subproblem} An instance of a subproblem is given by $i\in\{1,\dots,\Gamma\}$ and numbers $w'\leqslant w''\in \{0,1,\dots,n^{3}\}$. A solution is a schedule that starts at time $t_i$ and completes before time $t_{i+1}$ and completes a total weight of at least $w''$. let $C^{w}$ be the moment that the schedule completes a total weight of at least $w$. The objective is to minimize $\sum_{w=w'+1}^{w''}C^{w}$. 
Note that an instance $(i,w',w'')$ may not be feasible. For any feasible instance, let $\SUB_i(w',w'')$ be its optimal value.

\begin{definition}
An $(\alpha,\beta)$ approximation algorithm for the subproblem is an algorithm that finds for any \emph{feasible} instance  $(i,w',w'')$ a schedule that does not start before time  $\alpha t_i$ and ends before  time $\alpha t_{i+1}$, completes a total weight of at least  $w''$, and for which $\sum_{w=w'+1}^{w''}C^{w}\leqslant \alpha\beta\SUB_i(w',w'')$.
\end{definition}

Note that the total weight $w''$ is not approximated in the definition above. For example, completing a total weight of $(1-\eps)w''$ is not sufficient to obtain a PTAS.

Assume we have an  $(\alpha,\beta)$-approximation algorithm $\ALG$ for the subproblem. Let $\Alg_i(w',w'')$ be the value returned by the algorithm for instance $(i,w',w'')$ and let it be infinite if no solution was found. 
For any sequence of integers $0\le\hat{w}_1\leqslant \dots \le\hat{w}_{\Gamma}=W$ we get a pseudo schedule  of total weighted completion time 
\begin{equation}\label{eq:total_ALGschedule}
\sum_{i=1}^{\Gamma}\alg_i(\hat{w}_{i-1},\hat{w}_i)\leqslant \alpha\beta \sum_{i=1}^{\Gamma}\SUB_i(\hat{w}_{i-1},\hat{w}_i)
\end{equation}
by concatenating the schedules $\alg_i(\hat{w}_{i-1},\hat{w}_i)$.
Minimizing the left side of~\eqref{eq:total_ALGschedule} over all values $0\le\hat{w}_1\leqslant \dots \le\hat{w}_{\Gamma}=W$ is easy since they form a non-decreasing sequence and the minimum can be computed by a simple dynamic program similarl to what was done for the TRP. Let the values $\hat{w}_i$ minimize the left side of~\eqref{eq:total_ALGschedule} and let $w_i$ be the total weight in the partial solution $\OPT_i$. Then, algorithm $\ALG$ finds a solution of total weighted completion time at most  
\begin{eqnarray*}
\alpha\beta \sum_{i=1}^{\Gamma}\SUB_i(\hat{w}_{i-1},\hat{w}_i)&\le& \alpha\beta \sum_{i=1}^{\Gamma}\SUB_i(w_{i-1},w_i)\\
&\leqslant & \alpha\beta\OPT'\le\alpha\beta(1+\eps)\OPT.
\end{eqnarray*}
The number of subproblems is $O(\Gamma n^6)$ and given all approximate values, the optimal values $\hat{w}_i$
can be computes in $O(\Gamma n^6)$ time. Further, the number of choices for $h_0$ is $K$ (See Equation~\ref{eq:h0}). Hence, it takes  $O(K\Gamma n^6)=\OO(n^6\log n)$ calls to the approximation algorithm for the subproblem to get an $\alpha\beta(1+\eps)$-approximation for our scheduling problem. \bigskip

\subsection{Approximating the subproblem.}\label{sec:approx sub} We show how to get a $(1+\eps,1+\eps)$-approximation for the subproblem. In this section, we fix an arbitrary subproblem with parameters $i,w',w''$ and fix an optimal solution $\SUB^*$.
Again, the first step is to partition the interval from $t_i$ till $t_{i+1}$ into $K$ parts that we shall denote as \emph{slots}. As before (Equation~\eqref{eq:t_ij}), let
\[
t_i^{(h)}=(1+\eps)^h t_i, \text{ for } h=0,\dots, K.
\]
From now, the approach will differ from what we did for the TRP. The general idea is as follows.
Since the number of slots in a subschedule is a constant $K$, and all weights and processing times are polynomially bounded, we can afford to guess a lot of information about $\SUB^*$. We shall do this in such a way that the remaining jobs can be scheduled greedily.
For the ease of analysis, we extend $\SUB^*$ by putting all unscheduled jobs at the end. We say that they are scheduled in a virtual slot $K+1$.
Now, for each job $j$ we guess a set of possible slots $S_{j}\subseteq\{1,2,\dots,K+1\}$ with the following properties:
\begin{itemize}
\item[(P1)] Any job $j$ in $\SUB^*$ completes in some slot in $S_{j}$. (It may start in an earlier slot though.)
\item[(P2)] If $j_1\prec j_2$ then $\max (S_{j_1})\leqslant \min (S_{j_2})$.
\end{itemize}

The first property is easily satisfied. For example, if we let $S_{j}$ be the set of all $K+1$ slots for each $j$. The second property is implied by the interval order precedence constraints as we shall prove in Lemma~\ref{lem:interval-properties} below. After this lemma, we prove that we get a PTAS for any class of precedence constraints for which  we can prove (P1) and (P2) and for which we may restrict to polynomially bounded weights and processing times. Roughly speaking, the consequence of (P1) and (P2) is that we only need to deal with precedence constraints within a slot. However, within a slot any order of the jobs that satisfies the precedence constraints is fine since all completion times are within a factor $(1+\eps)$.

\begin{lemma}\label{lem:interval-properties} For interval orders, we can guess sets $S_{j}$ for $j=1,2,\dots,n$, that satisfy  properties (P1) and (P2).
\end{lemma}
\begin{proof}
Let $[l_j,r_j]$ be the interval for job $j$ in the interval order. As noted, we may assume that the $2n$ values $l_j,r_j$ are all different. For any $h\in \{1,\dots,K+1\}$, let $J^{h}$ be the set of jobs that complete in slot $h$ in $\SUB^*$. Note that $J^{h}$ may be empty. For each non-empty set, guess the job $j^{h}$ with the largest value $l_{j^{h}}$, i.e., $l_{j^{h}}=\max \{l_{j}\mid j\in J^{h}\}$ and define $S_{j^{h}}=\{h\}$. (Note that there are $n^{O(K)}$ possible guesses.)  
For any other job, the set $S_{j}$ is defined as the unique maximal subset of $\{1,\dots,K+1\}$  that satisfies the following four necessary conditions.
\begin{itemize}
\item[(a)] If for some $h$, the guess was $J^{h}=\emptyset$, then $h\notin S_j$.
\item[(b)] If $j\prec j^{h}$ for some slot $h$, then $\max (S_j)\leqslant h$.
\item[(c)] If $j^{h}\prec  j$ for some slot $h$, then $\min(S_j)\geqslant h$.
\item[(d)] If $l_{j}>L^{h}$ for some $h$, then $h\notin S_{j}$.
\end{itemize} 
(In (c), one might replace $\geqslant h$ by $\geqslant h+1$ since $h\notin S_j$ follows from (d).)
Assume that we guessed all jobs $j^{h}$ correctly. Then, property (P1) follows directly since the conditions (a)--(d) are clearly necessary. To prove (P2) assume that $j_1\prec j_2$. We distinguish three cases:\\
Case 1:
 $j_2=j^h$ for some $j^h$. It follows from (b) that $\max (S_{j_1})\leqslant h = \min (S_{j_2})$, since $S_{j_2}=\{h\}$.\\
Case 2: $j_1=j^h$ for some $j^h$. It follows from (c) that $\min (S_{j_2})\geqslant h= \max (S_{j_1})$, since $S_{j_1}=\{h\}$.\\
Case 3: Now assume that $j_1,j_2\neq j^h$ for any $j^h$. Let $h=\min (S_{j_2})$. Then by (a), $J^{h}\neq\emptyset$. Then by (d), $l_{j_2}<l_{j^{h}}$. It follows form $j_1\prec j_2$ that $r_{j_1}<l_{j_2}<l_{j^{h}}$, Hence, $j_1\prec j^h$ and then (b) implies $\max (S_{j_1})\leqslant h=\min (S_{j_2})$.
\qed\end{proof}\bigskip

Assume from now on that we have sets $S_j$ satisfying (P1) and (P2).
\bigskip

\paragraph{Constructing the schedule} We will construct a $(1+\eps,1+\eps)$-approximate schedule $\sigma$. 
The construction is done as follows. First, we assign each job $j$ to some slot in $S_j$. Jobs that are not assigned to any of the first $K$ slots are implicitly assigned to the virtual slot $K+1$.
The slots $1,2,\dots,K$ are placed one after the other in this order, startingat time $(1+\eps)t_i$, and  within a slot the jobs are placed in any arbitrary order that satisfies the precedence constraints.  
By property (P2), the resulting schedule $\sigma$ is guaranteed to be feasible.
The word \emph{slot} is ambiguous here since the start and end time of slots in $\sigma$ are not fixed and do not match those of $\SUB^*$. We will show however that in the final schedule $\sigma$, the end time of slot $h$ is at most a factor $1+\eps$ larger than that of slot $h$ in $\SUB^*$.

Say that a job is \emph{large} if its processing time is at least $f(\eps) t_{i}$, where $f(\eps)$ is some function of $\eps$ to be specified later. Call it \emph{small} otherwise. Since there can only be a constant number (depending on $\eps$) of large jobs scheduled in $\SUB^*$ we
\begin{itemize}
\item  guess all large jobs together with the slot ($1,\dots,K$) in which they complete  in $\SUB^*$ and assign a job to slot $h$ in $\sigma$  if it completes in slot $h$  in $\SUB^*$.
\end{itemize}
It remains to assign the small jobs.
Note that there are at most $2^{K+1}$ different sets $S_j$. For any $S\subseteq\{1,\dots,K+1\}$, let $J_S=\{j\mid S_j=S \text{ and $j$ is small}\}$. 
\begin{itemize}
\item For every pair $(S,h)$, with $h\in \{1,\dots,K\}$ we guess the total processing time over all jobs $j\in J_S$ which complete in $\SUB^*$ in slot $h$. Let $P(S,h)$ be this value.
\end{itemize}
For each $S\subseteq \{1,\dots,K+1\}$ place the jobs in $J_S$ in non-decreasing order $w_j/p_j$ and do the following: 
\begin{itemize}
\item For slots $h=1$ to $K$, assign jobs from $J_S$ in order $w_j/p_j$ to slot $h$ until the total processing time of jobs from $J_S$ assigned to $h$ becomes at least $P(S,h)$ or until all jobs from $J_S$ are assigned.
\end{itemize}
Given this assignment of jobs to slots, we schedule jobs within each slot in an arbitrary order that satisfies the precedence constraints. Note that there are only $\OO(1)$ large jobs and we can guess all of them together with their slots. Also, the number of pairs $(S,h)$ is $2^{K+1}K=\OO(1)$ and for each pair, the number of possible values $P(S,h)$ is $O(n^3)$ since $p_j\leqslant n^2$ for all $j$. Hence, the total number of choices for the guesses is $n^{\OO(1)}$.  Let $\sigma$ be the schedule that follows from correct guesses about $\SUB^*$. 

\begin{lemma}
Schedule $\sigma$ is a $(1+\eps,1+\eps)$-approximation for the subproblem.
\end{lemma}
\begin{proof}
By property (P2) and since we scheduled jobs within a slot in an order satisfying the precedence constraint, the schedule is feasible. 

Next we show that slot $h$ in $\sigma$ ends before time $(1+\eps)t_i^{(h)}$. 
Let $P_{\sigma}(S,h)$ be the total processing time of jobs from $J_S$ which are assigned to slot $h$ in $\sigma$. Further, let $P(h)$ be the total processing time of jobs that complete in slot $h$ in $\SUB^*$ and let $P_{\sigma}(h)$ be the total processing time of jobs assigned to slot $h$ in $\sigma$.  
Remember that a job is small if its processing time is at most $f(\eps)t_i$. By the greedy assignment of small jobs we have that
\[ P_{\sigma}(S,h)\leqslant  P(S,h)+ f(\eps)t_i.\]
The number of possible sets $S$ is $2^{K+1}$. Now, take $f(\eps)=\eps^2 \cdot 2^{-(K+1)}$. Then,
\[ P_{\sigma}(h)\leqslant  P(h)+ 2^{K+1}f(\eps)t_i=P(h)+\eps^2t_i.\]
Slot $1$ in $\SUB^*$ has length $\eps t_i$ and the total processing time assigned to slot 1 is at most $P(1)+\eps^2 t_i\leqslant \eps t_i+\eps^2 t_i=(1+\eps)\eps t_i$.
Slot 1 is the smallest slot in $\SUB^*$. Hence, in general, the total time assigned to the first $h$ slots is at most $(1+\eps)$ times the length of the first $h$ slots in $\SUB^*$. That means, slot $h$ in $\sigma$ ends before time $(1+\eps)t_i+(1+\eps)(t_i^{(h)}-t_i)=(1+\eps)t_i^{(h)}$. 

Next, we prove the bound on the value of the schedule $\sigma$. Take arbitrary  $S\subseteq \{1,\dots,K+1\}$. If $\SUB^*$ completes a total weight  $w$ of jobs form $J_S$ by the end of slot $h$, then our schedule will have completed at least the same weight of jobs form $J_S$ by the end of slot $h$ too, since we scheduled the jobs in $w_j/p_j$ order.
For any $w\in \{1,\dots, w''\}$, let $C^{*w}$ be the time at which $\SUB^*$ completes a total weight of at least $w$. Consider arbitrary $w$ and assume that time $C^{*w}$ falls in slot $h$. Then our schedule completes a total weight of at least $w$ before the end time of slot $h$. Hence, before time $(1+\eps)t_i^{(h)}=(1+\eps)^2t_i^{(h-1)}<(1+\eps)^2C^{*w}$. In particular, this applies to any $w\in \{w'+1,\dots,w''\}$. Hence,
\[\ALG_i(w',w'')\leqslant (1+\eps)^2\SUB_i(w',w'').\]\qed
\end{proof}

The PTAS for interval ordered precedence constraints can easily be adjusted to deal with release dates. First, the release dates may be  rounded such that jobs are released at the beginning of slots. Next, the release date restrictions are added to the sets $S_j$ as defined in the proof of Lemma~\ref{lem:interval-properties}. The rest remains the same.

\providecommand{\bysame}{\leavevmode\hbox to3em{\hrulefill}\thinspace}
\providecommand{\MR}{\relax\ifhmode\unskip\space\fi MR }
\providecommand{\MRhref}[2]{%
  \href{http://www.ams.org/mathscinet-getitem?mr=#1}{#2}
}
\providecommand{\href}[2]{#2}

\end{document}